\documentclass{amsart}
\usepackage[left=3cm,bottom=3cm,right=3cm,top=3cm]{geometry}
\usepackage{amsmath}
\usepackage{amssymb}
\usepackage{mathrsfs}
\usepackage{graphicx}
\usepackage{mathrsfs}
\usepackage{color}
\usepackage[normalem]{ulem}

\numberwithin{equation}{section}
\newtheorem{lemma}{Lemma}[section]
\newtheorem{defi}{Definition}[section]

\newtheorem{theorem}{Theorem}[section]
\newtheorem{rem}{Remark}[section]

\begin{document}
\title[Uniqueness of static, isotropic low-pressure solutions of the EV-system]{Uniqueness of static, isotropic low-pressure solutions of the Einstein-Vlasov system}
\author{Tomohiro Harada \& Maximilian Thaller}

\address{
Tomohiro Harada\newline
Department of Physics, \newline
Rikkyo University, Toshima, Tokyo 171-8501, Japan \newline
harada@rikkyo.ac.jp
}

\address{
Maximilian Thaller \newline
Chalmers University of Technology, \newline
Department of Mathematical Sciences, \newline
412 96 Gothenburg, Sweden \newline
maxtha@chalmers.se
}

\maketitle

\begin{abstract}
In \cite{bs90} the authors prove a uniqueness theorem for static solutions of the Einstein-Euler system which applies to fluid models whose equation of state fulfills certain conditions. In this article it is shown that the result of \cite{bs90} can be applied to isotropic Vlasov matter if the gravitational potential well is shallow. To this end we first show how isotropic Vlasov matter can be described as a perfect fluid giving rise to a barotropic equation of state. This {\em Vlasov} equation of state is investigated and it is shown analytically that the requirements of the uniqueness theorem are met for shallow potential wells. Finally the regime of shallow gravitational potential is investigated by numerical means. An example for a unique static solution is constructed and it is compared to astrophysical objects like globular clusters. Finally we find numerical indications that solutions with deep potential wells are not unique.
\end{abstract}

\tableofcontents

\section{Introduction}
In \cite{bs90} R.~Beig and W.~Simon prove that a static solution of the Einstein-Euler system with matter quantities of compact support is spherically symmetric and uniquely determined by the surface potential of the fluid body if the matter fulfills certain assumptions. This is done by showing that a static solution of the Einstein-Euler system is isometric to a spherically symmetric solution if the assumptions of their theorem are satisfied. \par
This theorem thus implies two statements. First that a solutions satisfying the assumptions is spherically symmetric. Second, that it is uniquely determined by the surface potential. For perfect fluids with Newtonian gravity it is in fact already known for a long time that static solutions automatically are spherically symmetric \cite{l33}. In the framework of General Relativity it turned out to be surprisingly difficult to establish comparable results. An important step is the analysis of Masood-ul-Alam \cite{m88} published in 1988. A uniqueness result is given which however only applies to restricitve and somewhat unphysical assumptions on the equation of state. In the following years the analysis could be extended to different types of equations of state which are physically more relevant. See \cite{l92} for a chronological overview of results on symmetry and uniqueness. To our knowledge the question whether isotropic Vlasov matter, which can be seen as perfect fluid, is covered by these uniqueness results has not been addressed yet. \par
This article is concerned with the Einstein-Vlasov system and it is investigated under which assumptions the main theorem of \cite{bs90} can be applied to it. We have for example the question in mind if a globular cluster necessarily is spherically symmetric. Vlasov matter is a natural choice of matter model to describe this situation. Generally, Vlasov matter possesses features that cannot be described by a perfect fluid. In particular the momenta of the particles can be distributed anisotropically. It is known that there exist static, anisotropic solutions of the Einstein-Vlasov system which are not spherically symmetric \cite{akr11}. Isotropic solutions of the static Einstein-Vlasov system however resemble perfect fluid solutions much more. In the non-relativistic case it is known that isotropic static solutions, so called {\em steady states}, are necessarily spherically symmetric and unique in a certain sense. This is a byproduct of the method of proof of existence. It is shown that a steady state is a minimizer of the so called energy-Casimir functional. These minimizers then turn out to be spherically symmetric. See \cite{rein} for details. For the Einstein-Vlasov system similar methods have not yet been successfully applied. Some progress has been made under the assumption that the considered steady states are not very relativistic in a certain sense \cite{hr12}.  Existence of isotropic, spherically symmetric static solutions of the Einstein-Vlasov system has been established by other methods \cite{rr93}. \par
This picture, that not very relativistic static solutions of the Einstein-Vlasov system are necessarily spherically symmetric whereas it is unclear for highly relativistic ones, is confirmed in this article. We show that under the assumption of isotropy the energy momentum tensor of Vlasov matter is described by two functions $\varrho$ and $p$ which can be seen as energy density and pressure of a perfect fluid satisfying a barotropic equation of state. In terms of these functions $\varrho$ and $p$, an additional function $I$, defined in (\ref{eq_def_i}) below, can be introduced. By the uniqueness theorem of \cite{bs90} a static fluid solution is unique if $I\leq 0$. So we analyze the equation of state resulting from isotropic Vlasov matter and show analytically that $I \leq 0$ in the regime of relatively low pressure. The applied method is very robust so that the effect of different choices of ansatz functions for the particle distribution can be studied. We are able to characterize a large class of solutions which will be unique in the {\em low-pressure regime}. At the same time we are able to give criteria on the particle distribution function revealing that the resulting equation of state is not compatible with the uniqueness result of \cite{bs90}.\par
In the last part of the paper the low-pressure regime is investigated further by numerical means. We calculate explicit examples of static solutions of the Einstein-Vlasov system in spherical symmetry with low pressures. Due to the analytical result of this article we conclude that these solutions are unique. Further, we discuss numerical indications that in the high-pressure regime static solutions are not unique, namely we find different spherically symmetric solutions which have the same surface potential. \par
For an isotropic particle distribution function of the Vlasov matter, the maximum pressure and the {\it concentration parameter} $\Gamma$ of a spherically symmetric solution are correlated. Therefore static solutions of the Einstein-Vlasov system with a low concentration parameter $\Gamma$ are necessarily spherically symmetric. We calculate the maximum concentration parameter in the low pressure regime for an example family of particle distribution functions and set this into relation to observational values of existing astrophysical objects. At the example of the numerically calculated family of solutions we will see that solutions with a concentration parameter comparable to neutron stars are not in the low pressure regime and the main theorem of \cite{bs90} cannot be applied. Stars or globular clusters however are in this regime.

\subsection*{Acknowledgments}
We thank J\'er\'emie Joudioux for mentioning \cite{bs90} to us, even though he might have had completely different questions in mind. T.H.~is grateful to T.~Shiromizu for fruitful discussion. M.T.~thanks H\aa{}kan Andr\'easson and Simone Calogero for helpful discussions. Moreover M.T.~is grateful the Department of Theoretical Physics at Rikkyo University for hospitality during a research stay June -- August 2017, as well as the Japan Society for the Promotion of Science (JSPS) and the Swedish Foundation for International Cooperation in Research and Higher Education (STINT) for financial support. This work was partially supported by JSPS KAKENHI Grant No. JP26400282 (T.H.).

\section{Preliminaries}

\subsection{Fluid models}

In this article units in which $G=c=1$ are used. Moreover we use the Einstein summation convention. Greek indices run from 0 to 3 and Latin indices run from 1 to 3. Let $\mathscr M$ be a four dimensional manifold equipped with a Lorentzian metric $g$. We assume that $(\mathscr M, g)$ is a static space-time. This means that there exists a three dimensional manifold $\Sigma$ such that $\mathscr M \cong \mathbb R \times \Sigma$ and there exist coordinates $t$, $x^1$, $x^2$, $x^3$ such that the metric $g$ can be written as
\begin{equation} \label{eq_metric}
g = -V^2\left(x^1,x^2,x^3\right) \mathrm dt^2 + \gamma_{ab}\left(x^1, x^2, x^3\right) \mathrm dx^a \mathrm dx^b,
\end{equation}
where $V\in C^1\left(\Sigma; \mathbb R\right)$ and the Riemannian metric $\gamma_{ab}$ is the restriction of $g$ to $\Sigma$. \par
Assume that on $\Sigma$ we have two functions $\varrho,p\in C^2(\Sigma; \mathbb R)$. Assume furthermore that the support of these functions is compact and let $Q\subset \Sigma$ be an open set such that $\bar Q = \mathrm{supp}(\varrho) \cup \mathrm{supp}(p)$. If the functions $\varrho$ and $p$ are related by a barotropic equation of state, i.e.~$\varrho=\varrho(p)$, $\mathrm d\varrho/\mathrm d p \geq 0$, then they give rise to a {\em fluid model} which we define as in \cite{bs90}.

\begin{defi} (Fluid model) \label{def_fluid_model} \\
Let $\varrho,p\in C^2(\Sigma; \mathbb R)$ satisfy a barotropic equation of state, i.e.~$\varrho=\varrho(p)$, $\mathrm d\varrho/\mathrm d p \geq 0$. A corresponding fluid model is a triple $(\Sigma, \gamma_{ab}, V)$, where $\Sigma$ is a three dimensional Riemannian manifold endowed with the metric $\gamma_{ab}$, and $V\in C^1(\Sigma; \mathbb R)$ such that the Einstein equations 
\begin{subequations}
\begin{eqnarray} 
R_{ab} &=& \frac{1}{V} D_a D_b V + 4\pi (\varrho-p) g_{ab}, \label{einst_eq_1} \\
\Delta V &=& 4\pi V (\varrho + 3p), \label{einst_eq_2}
\end{eqnarray}
\end{subequations}
hold. Here $D_a$ is the covariant derivative formed from $\gamma_{ab}$, $\Delta = \gamma^{ab} D_a D_b$, and $R_{ab}$ is the Ricci tensor formed from $\gamma_{ab}$.
\end{defi}

Next we define the quantity
\begin{equation} \label{eq_def_i}
I = \frac 1 5 \kappa^2 + 2\kappa + (\varrho + p) \frac{\mathrm d\kappa}{\mathrm dp},\qquad \mathrm{where}\quad \kappa = \frac{\varrho + p}{\varrho + 3p}\frac{\mathrm d\varrho}{\mathrm dp}.
\end{equation}
We review the main theorem of \cite{bs90} that this article relies on.

\begin{theorem} (Beig \& Simon, 1990)\label{theo_bs} \\
Assume we are given a static perfect fluid model $(\Sigma, \gamma_{ab}, V)$ with equation of state satisfying $I\leq 0$, and a spherically symmetric solution $(\mathbb R^3, ^0\gamma_{ab}, ^0V)$. Then the given model and the spherically symmetric solution are isometric.
\end{theorem}

The function $^0V$ turns, by the theorem, out to be the same as $V$ and it depends on the radial coordinate only. Moreover it is monotonically increasing. So the theorem is a uniqueness result in the sense that the value of $V$ at the boundary $\partial Q$ of the fluid body uniquely determines the space-time (for a fixed equation of state). The statement that all static solutions, which have an equation of state such that the assumptions of the theorem are satisfied, are spherically symmetric is an immediate consequence. We call this value of $V$ the {\em surface potential}. Later in the analysis (Section \ref{sect_main}), a cut-off energy $E_0$ will be introduced. Its value will be exactly this surface potential.

\subsection{Vlasov matter}

We consider an ensemble of particles in $\mathscr M$ which move along timelike geodesics. Let $x(\sigma) = \left(x^0(\sigma), x^1(\sigma),  x^2(\sigma),  x^3(\sigma)\right)$ be a future-directed geodesic and let
\begin{equation}
p^\mu(\sigma) := \frac{\mathrm d}{\mathrm d\sigma} x^\mu(\sigma)
\end{equation}
be the canonical momenta. Then $p^\mu(\sigma)$ fulfills the geodesic equation
\begin{equation}
\frac{d p^{\mu}}{d\sigma}-\Gamma^{\mu}_{\nu\lambda}p^{\nu}p^{\lambda}=0, \quad \mu=0,\dots,3.
\end{equation}
The rest mass $m$ of the particle following the geodesic $x^\mu(\sigma)$ is defined by
\begin{equation}  \label{eq_m_const}
m^2 = -g_{\mu\nu}(x(\sigma))\,p^{\mu}(\sigma)p^{\nu}(\sigma).
\end{equation}
It can be shown \cite{sz14} that the rest mass $m$ stays constant along the geodesic $x^\mu(\sigma)$. We note that the parameter $\sigma$ is proper time if and only if $m=1$. Otherwise we have $d\sigma=d\tau/m$ so that $p^{\mu}=dx^{\mu}/d\sigma$. \par
The mass shell $P_m$ is defined to be
\begin{equation} \label{def_mass_shell}
P_m = \{ (x,p)\in T\mathscr M\,:\,g_{\mu\nu}(x)p^\mu p^\nu = -m^2,\;p\;\mathrm{future}\,\mathrm{directed} \}.
\end{equation}
The mass shell is a seven dimensional submanifold of $T \mathscr M$ containing the lifts, to the tangent bundle $T\mathscr M$, of the future directed geodesics in $\mathscr M$, corresponding to particles with rest mass $m$. On $P_m$ we define the distribution function $f\in C^1\left(P_m;\mathbb R\right)$ of the particles with rest mass $m$ which satisfies the 
Vlasov equation,
\begin{equation} \label{vlasov_eq}
p^{\mu} \frac{\partial}{\partial x^{\mu}}f - \Gamma^{i}_{\nu  \lambda} p^{\nu} p^{\lambda} \frac{\partial}{\partial p^{i}}f = 0.
\end{equation} 
We will write $(p^1,p^2,p^3)$ as ${\bf p}$ and $(x^1,x^2,x^3)$ as ${\bf x}$ etc. The stress-energy tensor for $m>0$ is given by 
\begin{equation} \label{def_t}
T^{\mu\nu}(x^\sigma) = \frac{1}{m} \int_{P_{(m,x)}} f(x^\sigma,{\bf p}) p^{\mu} p^{\nu} \, \mu_{P_{(m,x)}},
\end{equation}
where $P_{(m,x)}$ is the fiber of the mass shell $P_m$ which is a submanifold of $T_x\mathscr M$, and $\mu_{P_{(m,x)}}$ is the volume form on $P_{(m,x)}$. For the massless case, $m=0$, a formula can be obtained by a continuity argument, cf.~(\ref{eq_for_t}) below.
We should note that $T^{\mu\nu}$ behaves as a covariant tensor in spite of the apparently three-dimensional volume integral. We also define the particle number current as 
\begin{equation} \label{def_n}
 N^{\mu}(x^\sigma)= \frac{1}{m} \int_{P_{(m,x)}} f(x^\sigma,{\bf p}) p^{\mu}\, \mu_{P_{(m,x)}}.
\end{equation}
One can show the conservation laws for the above quantities, $\nabla_{\mu}T^{\mu\nu} = 0$, $\nu=0,\dots,3$, and $\nabla_{\mu} N^{\mu} = 0 $. See \cite{a11} for a review article on the Einstein-Vlasov system and \cite{sz14} for more details on the geometric set-up.

\section{Tetrad description}

A solution $(\mathscr M, g, f)$ of the Einstein-Vlasov system is a Lorentzian metric $g$, defined on the manifold $\mathscr M$, such that the Einstein equations, $G_{\mu\nu} = 8\pi T_{\mu\nu}$, $\mu,\nu=0,\dots,3$, are satisfied where the Einstein tensor $G_{\mu\nu}$ is calculated from the metric $g$. Further, $f$ is a particle distribution function satisfying the Vlasov equation (\ref{vlasov_eq}) and giving via (\ref{def_t}) rise to the energy momentum tensor $T_\mu\nu$ on the right hand side of the Einstein equations. \par
In this section we show that a solution of the Einstein-Vlasov system with isotropic particle distribution function is a fluid model in the sense of Definition \ref{def_fluid_model}. To this end we first show in Lemma \ref{fluid_lem} how to express the energy momentum tensor (\ref{def_t}) of Vlasov matter in the form of a perfect fluid. In this section, by $\eta^{AB}$ we denote the components of the Minkowski metric, i.e.~$\eta^{00}=-1$, $\eta^{0I}=\eta^{I0}=0$ and $\eta^{IJ}=\delta^{IJ}$ for $I,J=1,2,3$. \par
The tangent bundle $T\mathscr M$ of $\mathscr M$ can be seen as eight dimensional manifold which is naturally equipped with the coordinates $x^\mu, p^\nu$, $\mu,\nu=0,\dots, 3$, where $p^\mu$ is the canonical momentum corresponding to the coordinate $x^\mu$. However, to formulate isotropic distribution, it is useful to introduce an orthonormal basis $\{e_{(A)}\}$ ($A=0,1,2,3$) for the tangent bundle, i.e.~$g(e_{(A)},e_{(B)}) = \eta_{AB}$, which we call a tetrad basis. We define $v^{(A)}$ as the components of the vector $p^{\mu}\partial_{x^{\mu}}$ with respect to this tetrad frame $\{e_{(A)}\}$ so that $p^{\mu}\partial_{x^{\mu}}=v^{(A)}e_{(A)}$, where and hereafter we adopt Einstein convention for summation with respect also to the tetrad components with $A,B=0,1,2,3$ and $I,J,K,L=1,2,3$. Let $e_{(A)}^\mu$ be the coefficients of this frame, i.e.~$e_{(A)} = e_{(A)}^\mu \partial_{x^\mu}$. Then we have the identity
\begin{equation} \label{frame_metr}
g^{\mu\nu} = \eta^{AB} e^{\mu}_{(A)} e^{\nu}_{(B)}.
\end{equation}
We use the notation $x:=(x^{0},x^{1},x^{2},x^{3})$ and ${\bf v}:=(v^{(1)},v^{(2)},v^{(3)})$ for abbreviation.

\begin{defi} \label{def_iso}
A matter distribution function $f\in C^1(P_m;\mathbb R)$ is called isotropic if there exist a tetrad basis $\{e_{(A)}\}$ ($A=0,1,2,3$) and a function $F: \mathbb R^4 \times \mathbb R_+ \to \mathbb R_+$ such that
\begin{equation}
f(x,{\bf p}) = F\left( x, v \right), \quad v^{2}= \delta_{IJ} v^{(I)} v^{(J)},
\end{equation}
for all $(x,{\bf v}) \in P_m$. We call $F$ an isotropic ansatz function.
\end{defi}

\begin{lemma} \label{fluid_lem}
Let $(\mathscr M, g)$ be a four dimensional Lorentzian space-time and $P_m$ the corresponding mass shell for $m \geq 0$, equipped with the coordinates $x$ and ${\bf p}$ as described above. Further, let $F: \mathbb R^4 \times \mathbb R_+ \to \mathbb R_+$ be an isotropic ansatz function for the matter distribution function $f$, satisfying
\begin{equation}
F \left(x, v \right)=\mathcal O(v^{-4-\epsilon}),\quad \epsilon>0, 
\end{equation}
for $v\in\mathbb R_+$. Further let
\begin{align}
\varrho(x) &:= 4\pi \int_0^\infty F \left(x, v\right) v^2 \sqrt{m^2+v^2} \, \mathrm dv, \label{def_rho} \\
p(x) &:= \frac{4\pi}{3} \int_0^\infty F \left(x, v\right) \frac{v^4}{\sqrt{m^2 + v^2}} \, \mathrm dv. \label{def_p}
\end{align}
Then there exists a unit timelike vector field $u$ such that the energy momentum tensor $T^{\mu\nu}$ defined in (\ref{def_t}) takes the form
\begin{equation} \label{eq_t_fluid}
T^{\mu\nu} = \varrho u^\mu u^\nu + p \left(u^\mu u^\nu + g^{\mu\nu}\right).
\end{equation}
\end{lemma}

\begin{proof}
First we express the components $T^{\mu\nu}(x)$ of the energy momentum tensor, defined in (\ref{def_t}), in terms of ${\bf v}$ in the tangent space $T_x\mathscr M$. To this end we calculate the volume form $\mu_{P_{(m,x)}}$ of the fibre $P_{(m,x)}$ which is a submanifold of $T_x\mathscr M$. Note that the mass shell condition reads
\begin{equation} \label{eq_msc_v}
v^{(0)} =\sqrt{m^{2}+v^{2}}.
\end{equation}
The tangent space $T_x \mathscr M$ can be seen as a four dimensional manifold endowed with the Minkwoski metric in the coordinates $v^{(A)}$, $A=0,1,2,3$. For $m>0$, using (\ref{eq_msc_v}) we calculate the restriction $\rho$ of the Minkowski metric to $P_{(m,x)}$. We obtain
\begin{equation}
\rho = \delta_{IJ} \mathrm dv^{(I)} \mathrm dv^{(J)} - \delta_{IK} \delta_{JL} \frac{v^{(I)} v^{(J)}}{\left(v^{(0)}\right)^2} \mathrm dv^{(K)} \mathrm dv^{(L)}.
\end{equation}
Thus we have
\begin{equation} \label{rho_deg}
\mu_{P_{(m,x)}} = \sqrt{\left|\det(\rho)\right|} \mathrm dv^{(1)} \mathrm dv^{(2)} \mathrm dv^{(3)} = \frac{m}{v^{(0)}} \mathrm dv^{(1)} \mathrm dv^{(2)} \mathrm dv^{(3)}.
\end{equation}
This formula is valid in the massive case. In the massless case, however, the metric (\ref{rho_deg}) is degenerate. Since $p^\mu = e^\mu_{(A)} v^{(A)}$, a straightforward calculation using (\ref{frame_metr}) yields the formula
\begin{equation} \label{eq_for_t}
T^{\mu\nu}(x) = e^\mu_{(A)}\big|_x e^\nu_{(B)}\big|_x  \int_{\mathbb R^3} f(x,v) v^{(A)} v^{(B)} \frac{1}{v^{(0)}} \mathrm dv^{(1)} \mathrm dv^{(2)} \mathrm dv^{(3)}
\end{equation}
for the energy momentum tensor. The notation $e^\mu_{(A)}|_x$ denotes that the vector $e^\mu_{(A)}$ is evaluated at the space-time point $x$. By a continuity argument this formula is also valid in the massless case, $m=0$. Then the formulas (\ref{def_rho}) and (\ref{def_p}) for $\varrho$ and $p$ yield
\begin{equation}
T^{\mu\nu}(x) = e^\mu_{(0)} \big|_x e^\nu_{(0)} \big|_x \varrho(x) + \left(e^\mu_{(0)} \big|_x e^\nu_{(0)} \big|_x + g^{\mu\nu}(x) \right) p(x).
\end{equation}
Finally we check that
\begin{equation} \label{eq_u}
u(x) := e_{(0)} \big|_x=e^\mu_{(0)} \big|_x\,\partial_{x^\mu}
\end{equation}
is a timelike unit vector field. 
Using $g\left(e_{(0)},e_{(0)}\right)=-1$ one establishes all desired properties of $u$ and the lemma is shown.
\end{proof}

For the static case, we can identify $u=e_{(0)}=e^{0}_{(0)}\partial_{t}$ and write $F(x,v)=F({\bf x},v)$, where ${\bf x}:=(x^{1},x^{2},x^{3})$. Taking an energy momentum tensor of the form (\ref{eq_t_fluid}) as right hand side to Einstein's equations $G_{\mu\nu}=8\pi T_{\mu\nu}$ and calculating the Einstein tensor $G_{\mu\nu}$ on the left hand side from a metric of the form (\ref{eq_metric}) one finds the system (\ref{einst_eq_1})--(\ref{einst_eq_2}) in Definition \ref{def_fluid_model} of a fluid model. \par
By inspection of the formulas (\ref{for_rho}) and (\ref{for_p}) below one notices that both $\varrho$ and $p$ are decreasing with respect to $V$. Thus $p(V)$ can be inverted and we write $V(p)=p^{-1}(p)$. Further, $\varrho$ and $p$ fulfill a barotropic equations of state. By slight abuse of notation we write $\varrho(p) = \varrho(V(p))$. In conclusion, in this section we have seen that an isotropic static solution of the Einstein-Vlasov system is a fluid model in the sense of Definition \ref{def_fluid_model}.

\section{Main result} \label{sect_main}

In this section we set $m=1$ and we assume that the metric $g$ is of the form (\ref{eq_metric}). Since $(\mathscr M,g)$ is a static space-time by assumption, the timelike vector field $\partial_t$ is Killing. Then the quantity
\begin{equation} \label{def_e}
E = -g\left( p^\mu, \partial_t \right) = V({\bf x})^2 p^0 = V({\bf x}) \sqrt{1+v^2}, \qquad {\bf x}\in \mathbb R^3
\end{equation}
is conserved along the geodesics with tangent vector $p^\mu$. Note that the mass shell condition (\ref{def_mass_shell}) and the frame components have been used in the formula (\ref{def_e}) for the particle energy $E$, and recall that we denote $v^{2}=\delta_{IJ}v^{(I)} v^{(J)}$. This implies that if $f$ depends on $x$ and $v$ only indirectly via $E$ it satisfies the Vlasov equation. Henceforth we assume that $f$ is a function of $E$. We denote this function by $\Phi$. Moreover, we assume that there exists a cut-off energy $E_0 > 0$. This means that it is assumed that no particle has energy $E$ larger than this value. In other words $\Phi(E) = 0$ if $E > E_0$. For technical reasons we introduce the function $\phi : (-\infty,1]\to\mathbb R$, vanishing on $(-\infty,0)$, so that we can write 
\begin{equation}
F({\bf x}, v)= \Phi(E({\bf x},v)) =: \phi\left(1-\frac{E({\bf x},v)}{E_0}\right).
\end{equation}
Since $E$ depends only on the absolute value $v$ of ${\bf v}$, $\Phi$ is an isotropic particle distribution function, cf.~Definition \ref{def_iso}. \par
From a change of variables in the integrals (\ref{def_rho}) and (\ref{def_p}), given by
\begin{equation}
v=\sqrt{\left(\frac{E}{V}\right)^{2}-1},\quad dv =\frac{\frac{E}{V}}{\sqrt{\left(\frac{E}{V}\right)^{2}-1}}\frac{dE}{V},
\end{equation}
we find the formulas
\begin{align}
\rho &= 4\pi \frac{1}{V^{3}}\int_{V}^{E_0} \phi\left(1-\frac{E}{E_0}\right) E^{2}\sqrt{\left(\frac{E}{V}\right)^{2}-1} \, \mathrm dE, \label{for_rho} \\
p &= \frac{4\pi}{3}\frac{1}{V}\int_{V}^{E_0} \phi\left(1-\frac{E}{E_0}\right)
\left[\left(\frac{E}{V}\right)^{2}-1\right]^{3/2}\, \mathrm dE.  \label{for_p}
\end{align}
Since $\varrho(V) = p(V) = 0$ for all $V\geq E_0$ we call $E_0$ the {\em surface potential}. \par The results presented in this paragraph hold for ansatz functions $\phi$ that satisfy the following assumptions. We assume that $\phi : (-\infty,1] \to \mathbb R_+$ is an analytic function on $[0,1]$, $\phi(x)=0$ if $x<0$, and that
\begin{equation} \label{ass_phi}
\exists n\in\mathbb N\, : \qquad \phi'(0) = \dots = \phi^{(n-1)}(0) = 0, \quad \phi^{(n)}(0) > 0.
\end{equation}
In particular this implies, that the $n$-th derivative is discontinuous at $0$. The step function $\phi(x) = \chi_{[0,1]}(x)$, or functions like $\phi(x) = [x]_+$ and $e^x\chi_{[0,1]}(x)$ meet these requirements. Now we are ready to state the main result.

\begin{theorem} \label{main_the}
Let $E_0>0$ and let $\phi$ be an ansatz function satisfying the assumption (\ref{ass_phi}) with $n\leq 3$. Let $(\Sigma, \gamma_{ab}, V)$ be the fluid model corresponding to $\varrho$ and $p$ constructed from $\phi$ via (\ref{for_rho}) and (\ref{for_p}). Then, there exists $p_0>0$ such that if $\sup_{x\in\Sigma}p(x)\leq p_0$, the model is spherically symmetric and the unique solution of the Einstein-Vlasov system determined by $\phi$ and the surface potential $E_0$.
\end{theorem}

\begin{rem}
It turns out that the condition $n\leq 3$ is necessary if we use ansatz functions of the form (\ref{ass_phi}). If $n\geq 4$, so in particular the choice $\phi(x) = [x^4]_+$ will lead to $I(p) \to \infty$, as $p\to 0$ and the main Theorem of \cite{bs90} cannot be applied, i.e.~it neither can be deduced that the solution is unique nor that it is not unique.
\end{rem}

In the proof of Theorem \ref{main_the} the functions defined in (\ref{for_rho}) and (\ref{for_p}) play an important role. For this reason we first establish some technical lemmas to treat these functions, before we state the proof of Theorem \ref{main_the}. It is convenient to introduce for $\kappa\in \left\{\frac 3 2, \frac 1 2, -\frac 1 2\right\}$ the functions
\begin{equation} \label{eq_def_xi}
\xi_\kappa(V) = \frac{4\pi}{V} \int_V^{E_0} \phi\left(1-\frac{E}{E_0}\right) \left(\frac{E^2}{V^2}-1\right)^\kappa \mathrm dE, \quad V\in (0,E_0].
\end{equation}
Observe that
\begin{equation} \label{eq_rho_p_xi}
\varrho(V) = \xi_{\frac{3}{2}}(V) + \xi_{\frac{1}{2}}(V),\quad\mathrm{and}\quad p(V) = \frac{1}{3} \xi_{\frac{3}{2}}(V),
\end{equation}
if $V\in (0,E_0]$, whereas $\varrho(V) = p(V) = 0$ if $V > E_0$. We collect some facts about the functions $\xi_\kappa$ in the following lemmas.

\begin{lemma} \label{lem_ratio}
Let $\phi\in H^1((-\infty,1];\mathbb R_+)$. Then
\begin{equation} \label{eq_f_c_l}
\frac{\xi_{\kappa+1}(V)}{\xi_\kappa(V)} \to 0,\quad \mathrm{as}\;\, V \to E_0.
\end{equation}
Furthermore, if $\phi$ satisfies (\ref{ass_phi}), there exists $0 < V^* < E_0$ such that for all $V \in [V^*, E_0]$ we have
\begin{equation} \label{eq_s_c_l}
\frac{\xi_{\frac 1 2}(V)}{\xi_{-\frac 12}(V)} \leq \frac{101(E_0+V)(E_0-V)}{100 (3+2n)V^2},
\end{equation}
where $n$ is introduced in (\ref{ass_phi}).
\end{lemma}

\begin{proof}
We have
\begin{equation}
\xi_{\kappa+1}(V) = \frac{4\pi}{V} \int_{V}^{E_0} \Phi(E) \left(\frac{E^2}{V^2}-1\right)^{\kappa+1} \mathrm dE \leq \left(\frac{E_0^2}{V^2}-1\right) \xi_\kappa(V).
\end{equation}
So
\begin{equation}
\frac{\xi_{\kappa+1}(V)}{\xi_\kappa(V)} \leq \frac{E_0^2-V^2}{V^2}.
\end{equation}
The right hand side clearly goes to zero as $V\to E_0$ and the first claim (\ref{eq_f_c_l}) of the lemma is shown. \par
For the proof of the second claim (\ref{eq_s_c_l}) we first note
\begin{align}
\xi_{\frac 12}(V) &\leq \frac{4\pi \sqrt{E_0+V}}{V^2} \int_V^{E_0} \phi\left(\frac{E_0-E}{E_0}\right) \sqrt{E-V}\, \mathrm dE, \\
\xi_{-\frac 12}(V) &\geq  \frac{4\pi}{\sqrt{E_0+V}} \int_V^{E_0} \phi\left(\frac{E_0-E}{E_0}\right) \frac{1}{\sqrt{E-V}}\, \mathrm dE,
\end{align}
since $E\leq E_0$. Then we define
\begin{equation} \label{def_epsilon}
\epsilon := E_0-V
\end{equation}
and perform a change of variables in the integrals of $\xi_{\frac 12}$ and $\xi_{-\frac 12}$, given by
\begin{subequations}
\begin{align}
y &= E - V, \quad \Leftrightarrow \quad E = y+V, \\
\mathrm dE &= \mathrm dy, \\
y(V) &= 0, \quad y(E_0) = \epsilon.
\end{align}
\end{subequations}
This yields
\begin{align}
\xi_{\frac 12}(V) &\leq \frac{4\pi \sqrt{E_0+V}}{V^2} \int_0^{\epsilon} \phi\left(\frac{\epsilon-y}{E_0}\right) \sqrt{y} \, \mathrm dy, \\
\xi_{-\frac 12}(V) &\geq  \frac{4\pi}{\sqrt{E_0+V}} \int_0^{\epsilon} \phi\left(\frac{\epsilon-y}{E_0}\right) \frac{1}{\sqrt{y}} \, \mathrm dy. 
\end{align}
We consider
\begin{equation} \label{diff_bc}
\frac{V^2}{4\pi \sqrt{E_0+V}} \xi_{\frac 12}(V) - \frac{\sqrt{E_0+V}}{4\pi} X\epsilon \xi_{-\frac 12}(V) \leq \int_0^{\epsilon} \phi\left(\frac{\epsilon-y}{E_0}\right) \left(\sqrt{y} - X\epsilon \frac{1}{\sqrt{y}}\right) \, \mathrm dy,
\end{equation}
where we later will substitute
\begin{equation} \label{choice_x}
X = \frac{101}{100} \frac{1}{3+2n}.
\end{equation}
Recall that we assume that $\phi$ is analytic on $[0,1]$, and fulfills (\ref{ass_phi}). Then we can write for all $x \in[0,1]$
\begin{equation}
\phi(x) = \frac{\phi^{(n)}(0)}{n!} x^n + \frac{\phi^{(n+1)}(z_x)}{(n+1)!} x^{n+1}, \label{taylor11}
\end{equation}
where $z_x\in[0,1]$ is a number depending on $x$. Now the integral in (\ref{diff_bc}) can be calculated explicitly. This yields
\begin{align}
&\frac{V^2}{4\pi \sqrt{E_0+V}} \xi_{\frac 12}(V) - \frac{\sqrt{E_0+V}}{4\pi} X\epsilon \xi_{-\frac 12}(V) \\
&\leq  \frac{\phi^{(n)}(0)}{n! E_0^n} \epsilon^{n+\frac 32} \left(\sum_{i=0}^n\left(n\atop i\right) \frac{(-1)^i}{i+\frac 32} - X \sum_{j=0}^n\left(n\atop j\right) \frac{(-1)^j}{j+\frac 12}\right) \\
&\quad+ \frac{\left\|\phi^{(n+1)}\right\|_{L^\infty([0,1])}}{(n+1)! E_0^{n+1}} \epsilon^{n+\frac 52} \left| \sum_{k=0}^{n+1} \left(n+1\atop k\right) \frac{(-1)^k}{k+\frac 32} - X \sum_{\ell=0}^{n+1} \left({n+1} \atop \ell\right) \frac{(-1)^\ell}{\ell+\frac 12}\right|. \nonumber
\end{align}
We observe that the first term will be dominating for $\epsilon$ sufficiently small. By Lemma \ref{lem_help} below and the choice (\ref{choice_x}) for $X$ the first term is negative. Thus
\begin{align}
&\frac{V^2}{4\pi \sqrt{E_0+V}} \xi_{\frac 12}(V) - \frac{\sqrt{E_0+V}}{4\pi} X\epsilon \xi_{-\frac 12}(V) \leq 0 \\
\Leftrightarrow \quad  &\frac{ \xi_{\frac 12}(V)}{ \xi_{-\frac 12}(V)} \leq \frac{101(E_0+V)\epsilon}{100(3+2n)V^2}
\end{align}
and the second claim (\ref{eq_s_c_l}) is established. 
\end{proof}

\begin{lemma} \label{lem_help}
Let $n \geq 0$. Then we have
\begin{equation}
\sum_{i=0}^n \left(n \atop i\right) \frac{(-1)^i}{i+\frac 32} = \frac{1}{2n+3} \sum_{j=0}^n \left(n \atop j\right) \frac{(-1)^j}{j+\frac 12}.
\end{equation}
\end{lemma}

\begin{proof}
We notice that the sums are given by the hypergeometric functions,
\begin{align}
_2F_1\left(-n,\frac 32;\frac 52; z\right) &= \frac 32 \sum_{k=0}^n \left(n \atop k\right) \frac{(-1)^k}{k + \frac 32} z^k, \\
_2F_1\left(-n,\frac 12;\frac 32; z\right) &= \frac 12 \sum_{k=0}^n \left(n \atop k\right) \frac{(-1)^k}{k + \frac 12} z^k,
\end{align}
evaluated at $z=1$. For the values of these functions we have by the
 Chu-Vandermonde identity
\begin{align}
_2F_1\left(-n,\frac 32;\frac 52; 1\right) &= \frac{n!}{\frac 52 \cdots \frac{2n+1}{2} \frac{2n+3}{2}}, \\
_2F_1\left(-n,\frac 12;\frac 32; 1\right) &= \frac{n!}{\frac 32 \frac 52 \cdots \frac{2n+1}{2}}.
\end{align}
The assertion now follows.
\end{proof}

\begin{lemma} \label{lem_xi}
Let $E_0>0$ and $\Phi\in H^1([0,E_0];\mathbb R)$. Then the function $\xi_\kappa$, defined in (\ref{eq_def_xi}), is continuously differentiable for  $\kappa\in \left\{\frac 3 2, \frac 1 2\right\}$ and we have
\begin{align}
\xi_{\frac 32}'(V) &= -\frac 1 V \left(4 \xi_{\frac 32}(V) + 3\xi_{\frac 12}(V)\right), \label{xi1} \\
\xi_{\frac 12}'(V) &= -\frac 1 V \left(2 \xi_{\frac 12}(V) + \xi_{-\frac 12}(V) \right). \label{xi2}
\end{align}
\end{lemma}

\begin{proof}
Let $\Delta>0$ small. We consider
\begin{align}
&\frac{1}{\Delta} \left[ \int_{V - \Delta}^{E_0} \Phi(E) \left( \frac{E^2}{(V-\Delta)^2} - 1\right)^\kappa \mathrm dE - \int_V^{E_0} \Phi(E) \left( \frac{E^2}{V^2} - 1 \right)^\kappa \mathrm dE \right] \label{eq_ints}\\
&= \frac{1}{\Delta} \int_{V - \Delta}^{V} \Phi(E) \left( \frac{E^2}{(V - \Delta)^2} - 1 \right)^\kappa \mathrm dE \nonumber \\
&\qquad + \int_V^{E_0} \Phi(E) \frac{1}{\Delta} \left[ \left( \frac{E^2}{(V-\Delta)^2} - 1\right)^\kappa - \left(\frac{E^2}{V^2} - 1\right)^\kappa \right]\mathrm dE. \nonumber
\end{align}
If $\kappa >0$ then the first integral on the right hand side of (\ref{eq_ints}) goes to zero, as $\Delta\to 0$. So the derivative of the integral can be obtained by merely differentiating the integrand with respect to $V$.
\end{proof}

The arguments in the proof of Lemma \ref{lem_xi} cannot be applied to $\xi_{-\frac 12}$ since the first summand in (\ref{eq_ints}) does not converge to $0$ as $\Delta\to 0$ for $\kappa =-\frac 12$. Thus the analysis of the derivative of $\xi_{-\frac 12}$ requires a different approach. The derivative $\frac{\mathrm d}{\mathrm dV} \xi_{-\frac 12}(V)$ consists in two parts,
\begin{equation} \label{for_dxim12}
\frac{\mathrm d}{\mathrm dV} \xi_{-\frac 12}(V) = -\frac{1}{V} \left(\xi_{-\frac 12}(V) + \zeta(V)\right),
\end{equation}
where
\begin{equation} \label{formula_zeta}
\zeta(V) := -4\pi \frac{\mathrm d}{\mathrm dV} \left(\int_V^{E_0} \Phi(E) \left(\frac{E^2}{V^2}-1\right)^{-\frac 12}\, \mathrm dE \right).
\end{equation}
We proof the following lemma.

\begin{lemma} \label{lem_zeta}
If $n\leq 3$ ($n$ is introduced in the assumption (\ref{ass_phi}) on the ansatz function $\phi$), then for $\epsilon > 0$, (defined in (\ref{def_epsilon})) sufficiently small, we have
\begin{equation}
\left| \frac{\xi_{\frac 12}}{\xi_{-\frac 12}} \frac{\zeta}{\xi_{-\frac 12}}\right| < \frac 45.
\end{equation}
\end{lemma}

\begin{proof}
The first factor is already treated in Lemma \ref{lem_ratio}. So we focus on the second factor. First we calculate $\zeta(V)$. To this end we perform a change of variables in the integral (\ref{formula_zeta}), given by
\begin{subequations}
\begin{align}
x &= \frac{E-V}{E_0-V}, \quad \Leftrightarrow \quad E = x(E_0-V) + V, \\
\mathrm dE &= (E_0-V)\, \mathrm dx, \\
x(V) &= 0, \quad x(E_0) = 1.
\end{align}
\end{subequations}
This yields 
\begin{equation}
\zeta(V) = -4\pi \frac{\mathrm d}{\mathrm dV} \left((E_0-V) \int_0^1 \phi\left(\frac{(E_0-V)(1-x)}{E_0}\right) \left(\left(x\left(\frac{E_0}{V}-1\right) + 1\right)^2-1\right)^{-\frac 12}\, \mathrm d x \right).
\end{equation}
A straight forward calculation yields
\begin{equation}
\zeta(V) = \zeta_1(V) + \zeta_2(V),
\end{equation}
where
\begin{align}
\zeta_1(V) &= 4\pi \int_V^{E_0} \phi'\left(1-\frac{E}{E_0}\right) \frac{E_0-E}{E_0(E_0-V)}\left(\frac{E^2}{V^2}-1\right)^{-\frac{1}{2}} \mathrm dE, \\
\zeta_2(V) &= -\frac{4\pi}{\epsilon} \int_V^{E_0} \phi\left(1-\frac{E}{E_0}\right) \left(\frac{E^2}{V^2}-1\right)^{-\frac 1 2} \left(\frac{EE_0}{(E+V)V}-1\right) \mathrm dE. \label{eq_f_zeta2}
\end{align}
Consider the last term in the integral of $\zeta_2$. We notice that 
\begin{equation}
\frac{EE_0}{(E+V)V}-1 = -\frac 12 + \Gamma(\epsilon),
\end{equation}
where $\epsilon = E_0-V$ and $\Gamma(\epsilon)$ is a positive continuous function that satisfies $\Gamma(\epsilon)\to 0$, as $\epsilon \to 0$. So we can write for $\epsilon$ sufficiently small
\begin{equation} \label{est_zeta2}
\left| \zeta_2(V) \right| \leq \frac{V}{2\epsilon} \xi_{-\frac 12}(V).
\end{equation}
Next we consider $\zeta_1(V)$. We perform a change of variables, given by
\begin{subequations}
\begin{align}
\alpha &= \frac{E_0-E}{E_0}, \quad \Leftrightarrow \quad E = E_0(1-\alpha), \\
\mathrm dE &= -E_0\mathrm d\alpha, \\
\alpha(V) &= \frac{\epsilon}{E_0}, \quad \alpha(E_0)=0.
\end{align}
\end{subequations}
This yields
\begin{equation}
\zeta_1(V) =  4\pi\frac{VE_0}{\epsilon} \int_0^{\epsilon/E_0}  \phi'\left(\alpha\right) \alpha \left(E_0^2(1-\alpha)^2 - V^2\right)^{-\frac 12} \mathrm d\alpha
\end{equation}
and
\begin{equation}
\xi_{-\frac 12}(V) = 4\pi E_0 \int_0^{\epsilon/E_0} \phi(\alpha) \left(E_0^2(1-\alpha)^2 - V^2\right)^{-\frac 12} \mathrm d\alpha.
\end{equation}
Recall that we assume that $\phi$ is analytic on $[0,1]$ and fulfills (\ref{ass_phi}) for $n\geq 0$. This means $n$ is the lowest number such that $\phi^{(n)}(0) \neq 0$ ($\phi^{(n)}$ denotes the $n$-th derivative). Then we can write for all $\alpha\in[0,1]$
\begin{align}
\phi(\alpha) &= \frac{\phi^{(n)}(0)}{n!} \alpha^n + \frac{\phi^{(n+1)}(x_\alpha^1)}{(n+1)!}\alpha^{n+1}, \label{taylor1} \\
\phi'(\alpha) &= \frac{\phi^{(n)}(0)}{(n-1)!} \alpha^{n-1} +
 \frac{\phi^{(n+1)}(x_\alpha^2)}{n!}{\alpha^{n}}. \label{taylor2}
\end{align}
where $x_\alpha^1, x_\alpha^2 \in [0,1]$ in the remainder terms depend on $\alpha$.
We have 
\begin{multline}
\frac{\epsilon}{4\pi VE_0} \zeta_1(V) - \frac{1}{4\pi E_0} \frac{103 n}{102} \xi_{-\frac 12}(V) \\
= \int_0^{\epsilon/E_0} \left(\phi'(\alpha)\alpha - \frac{103n}{102} \phi(\alpha)\right) \left(E_0^2(1-\alpha)^2 - V^2\right)^{-\frac 12} \mathrm d\alpha.
\end{multline}
We consider $(\phi'(\alpha)\alpha - \frac{103n}{102} \phi(\alpha))$ separately, using (\ref{taylor1})--(\ref{taylor2}). We have for some $x_\alpha^3\in [0,1]$
\begin{align}
\phi'(\alpha)\alpha - \frac{103n}{102} \phi(\alpha) &= -\frac{\alpha^n}{102} \left(\frac{\phi^{(n)}(0)}{(n-1)!} + \alpha \frac{\phi^{(n+1)}(x_\alpha^3)}{n!} \right) \\
&\leq \frac{\alpha^n}{102} \left(\epsilon \frac{\left\|\phi^{(n+1)}\right\|_{L^\infty([0,1])}}{n!}- \frac{\phi^{(n)}(0)}{(n-1)!}\right).
\end{align}
This is negative for $\epsilon$ sufficiently small. We deduce that for $\epsilon$ sufficiently small
\begin{equation} \label{est_zeta1}
\frac{\epsilon}{4\pi VE_0} \zeta_1(V) - \frac{1}{4\pi E_0} \frac{103n}{102} \xi_{-\frac 12}(V) < 0 \quad \Leftrightarrow \quad  \frac{\zeta_1(V)}{\xi_{-\frac 12}(V)} < \frac{103 nV}{102 \epsilon} .
\end{equation}
Combining (\ref{eq_s_c_l}) from Lemma \ref{lem_ratio}, (\ref{est_zeta2}) and (\ref{est_zeta1}) we obtain
\begin{equation}
\left| \frac{\xi_{\frac 12}}{\xi_{-\frac 12}} \frac{\zeta}{\xi_{-\frac 12}}\right| \leq \frac{101 (E_0+V) \epsilon}{100(3+2n)V^2} \left(\frac{V}{2\epsilon} + \frac{103 nV}{102\epsilon} \right).
\end{equation}
We have $E_0+V = 2V + \epsilon$ and if $\epsilon$ is small enough, we can write $E_0+V < \frac{102}{101} 2V$. Then
\begin{equation}
\left| \frac{\xi_{\frac 12}}{\xi_{-\frac 12}} \frac{\zeta}{\xi_{-\frac 12}}\right| < \frac{102 + 206 n}{100(3+2n)} \leq \frac 45.
\end{equation}
The last step is obtained by substituting $n=3$.
\end{proof}

\begin{proof}[Proof of Theorem \ref{main_the}]
The proof of Theorem \ref{main_the} is an application of Theorem \ref{theo_bs}, the main theorem in \cite{bs90}. We show that a fluid model coming from an isotropic static solution of the Einstein-Vlasov system meets the assumptions of the main theorem in \cite{bs90}. In particular we will show that there exists $p_0>0$ such that $I(p)\leq 0$ for all $p\leq p_0$. \par
We use (\ref{eq_rho_p_xi}) and (\ref{xi1})--(\ref{xi2}) to calculate
\begin{align}
\frac{\mathrm d\varrho}{\mathrm d p} &= \frac{\varrho'}{p'} = 3\left(1 + \frac{2\xi_{\frac 1 2}+\xi_{-\frac 1 2}}{4\xi_{\frac 3 2} + 3\xi_{\frac 1 2}}\right), \\
\kappa &= \frac{4 \xi_{\frac 3 2} + 5 \xi_{\frac 1 2} + \xi_{-\frac 1 2}}{2 \xi_{\frac 3 2} + \xi_{\frac 1 2}}. \label{for_kappa}
\end{align}
A prime denotes the derivative with respect to $V$. A straight forward calculation using Lemma \ref{lem_xi} then yields
\begin{multline}
\frac{\mathrm d\kappa}{\mathrm d p} = \frac{\kappa'}{p'} = \frac{-3}{\left(2\xi_{\frac 3 2} + \xi_{\frac 1 2}\right)^2 \left(4\xi_{\frac 3 2} + 3 \xi_{\frac 1 2}\right)} \\ 
\times \left(12 \xi_{\frac 3 2} \xi_{\frac 1 2} + 18 \xi_{\frac 1 2}^2 + 2\xi_{\frac32}\xi_{-\frac12} + 8 \xi_{\frac 1 2} \xi_{-\frac 1 2} + \xi_{-\frac 1 2}^2+ \left(2\xi_{\frac 3 2} + \xi_{\frac 1 2}\right) V\xi_{-\frac 12}'\right).
\end{multline} 
Using the formula (\ref{for_dxim12}) for $\xi_{-\frac 12}'$ this becomes
\begin{equation}
\frac{\mathrm d\kappa}{\mathrm d p} = \frac{\kappa'}{p'} = \frac{-3}{\left(2\xi_{\frac 3 2} + \xi_{\frac 1 2}\right)^2 \left(4\xi_{\frac 3 2} + 3 \xi_{\frac 1 2}\right)} \left(12 \xi_{\frac 3 2} \xi_{\frac 1 2} + 18 \xi_{\frac 1 2}^2 + 7 \xi_{\frac 1 2} \xi_{-\frac 1 2} + \xi_{-\frac 1 2}^2 - \left(2\xi_{\frac 3 2} + \xi_{\frac 1 2}\right) \zeta \right).
\end{equation} 
Again by virtue of Lemma \ref{lem_ratio}, and formula (\ref{for_kappa}), there exists a function $\Gamma = \Gamma(V)$ which goes to zero as $V\to E_0$ such that we can write
\begin{align}
\frac 1 5 \kappa^2 &= \frac 1 5 \left(2\xi_{\frac 32} + \xi_{\frac 12}\right)^{-2} \left( \left(10+\Gamma \right) \xi_{\frac 1 2} \xi_{-\frac 1 2} + \xi_{-\frac 1 2}^2 \right) \label{eq_i_term2} \\
2\kappa &=  2 \left(2 \xi_{\frac 32} + \xi_{\frac 12}\right)^{-2} \left[\left(2\xi_{\frac 32} + \xi_{\frac 12}\right) \left(4 \xi_{\frac 3 2} + 5 \xi_{\frac 1 2} + \xi_{-\frac 1 2}\right)\right]  \label{eq_i_term3}\\
&= 2 \left(2\xi_{\frac 32} + \xi_{\frac 12}\right)^{-2} (1+\Gamma) \xi_{\frac 12} \xi_{-\frac{1}{2}} \nonumber
\end{align}
for $V$ close to $E_0$. Inserting (\ref{eq_i_term2}) and (\ref{eq_i_term3}) into the formula (\ref{eq_def_i}) for $I$ we derive
\begin{equation} 
I = -\left(2\xi_{\frac 32} + \xi_{\frac 12}\right)^{-2} \left((12 \xi_{\frac 32} \xi_{\frac 12} + 18 \xi_{\frac 12}^2 + (3-\Gamma) \xi_{\frac 12}\xi_{-\frac 12}+\frac 4 5 \xi_{-\frac 12}^2 - \left(2\xi_{\frac 3 2} + \xi_{\frac 1 2}\right) \zeta \right), \label{for_i_cancel}
\end{equation}
where $\Gamma$ is a different function than that in (\ref{eq_i_term2}) and (\ref{eq_i_term3}), but still with the property that $\Gamma \to 0$, as $V\to E_0$. Now, by Lemma \ref{lem_zeta}, and since $2\xi_{\frac 32}$ can be neglected compared to $\xi_{\frac 12}$ for $V$ close to $E_0$, the combination
\begin{equation}
 \frac 4 5 \xi_{-\frac 12}^2 - \left(2\xi_{\frac 3 2} + \xi_{\frac 1 2}\right) \zeta
\end{equation}
is positive for $V$ sufficiently close to $E_0$. Thus we can write
\begin{equation} \label{last_eq_i}
I \leq -\left(2\xi_{\frac 32} + \xi_{\frac 12}\right)^{-2} \left((12 \xi_{\frac 32} \xi_{\frac 12} + 18 \xi_{\frac 12}^2 + (3-\Gamma) \xi_{\frac 12}\xi_{-\frac 12}\right).
\end{equation}
By virtue of (\ref{eq_f_c_l}) in Lemma \ref{lem_ratio} we see that the last term, 
\begin{equation}
-\frac{(3-\Gamma) \xi_{\frac 12}\xi_{-\frac 12}}{\left(2\xi_{\frac 32} + \xi_{\frac 12}\right)^{2}} \to -\infty,\quad\mathrm{as}\;\, V\to E_0.
\end{equation}
Since all terms in (\ref{last_eq_i}) are negative, we have obtianed $I(V) \to -\infty$, as $V \to E_0$. Thus, by continuity there exists $V_0$ such that $I(V)\leq 0$ for all $V\geq V_0$. We set $p_0:=p(V_0)$.
\end{proof}

\section{Discussion}

\subsection{Limits of the Einstein-Vlasov system}

In the last section we showed that if a static, isotropic solution of the Einstein-Vlasov system has not too high pressure, it is the unique spherically symmetric solution to the prescribed ansatz function $\phi$ and surface potential $E_0$. \par
Now we consider different limits to get some insights when Theorem \ref{theo_bs} can be applied and in what situations the assumptions are not met. Fist we consider situations of high pressure, i.e.~the relativistic limit for massive particles, where $v^{2}\gg m^2$. We find with (\ref{def_rho}), (\ref{def_p})
\begin{equation}
\rho \simeq  4\pi \int_{0}^{\infty}v^{3}F(x,v) \mathrm dv, \qquad p \simeq \frac{1}{3}\rho.
\end{equation}
Hence the equation of state for radiation fluid, $\varrho(p) = 3p$, is recovered. For this equation of state $I(p)$ can easily be calculated. One obtains $I(p) = \frac{24}{5} > 0$ and Theorem \ref{theo_bs} cannot be applied. \par
In the non-relativistic limit for massive particles, where $v^{2}\ll m^2$, we find 
\begin{equation}
 \rho \simeq Nm= 4\pi \int_{0}^{\infty} v^{2}F(x,v) \mathrm dv, \qquad p\simeq \frac{1}{3}m\langle v^{2} \rangle N, 
\end{equation}
where 
\begin{equation}
 \langle v^{2} \rangle :=\frac{\int_{0}^{\infty}v^{2} \mathrm dv \frac{\mathrm dN}{\mathrm dv}}{\int_{0}^{\infty} \frac{\mathrm dN}{\mathrm dv} \mathrm dv}
\end{equation}
and $N$ is defined by $N^\mu =N u^\mu$ where $N^\mu$ is the particle number current defined in (\ref{def_n}) and $u^\mu$ is the four velocity of the fluid, given in (\ref{eq_u}). Thus, the equation of state for non-relativistic ideal gas, $\varrho(p) = \frac{3}{\langle v^2\rangle} p$, is  recovered. For this equation of state, with $\langle v^{2}\rangle$ constant, corresponding to isothermal gas, we obtain $I>0$, independent of $p$, as well.\par
For the massless case, where $m=0$, we can explicitly see
\begin{equation}
\rho = 4\pi \int_{0}^{\infty} v^{3}F(x,v) \mathrm dv, \qquad p=  \frac{1}{3}\rho,
\end{equation}
and hence $p=\rho/3$ is recovered. We should note that thermal equilibrium is not necessary for $p=\rho/3$. In other words, if we assume isotropy in the momentum space in the Vlasov system of massless particles, it necessarily reduces to the perfect fluid system with the equation of state $p=\rho/3$. We have already seen that for this equation of state one always has $I>0$, independently of $\phi$ and $E_0$. \par
Here, it is instructive to derive the equation of state in the low-pressure regime for the massive case in the context of Theorem~\ref{main_the}. For simplicity, in reference to (\ref{ass_phi}), we assume that $\Phi(E)=\phi(1-E/E_{0})$ has a cut off, i.e., $\phi(x)=0$ for $x<0$ and $\lim_{x\to +0}x^{-n}\phi(x) = C>0$ for $n\ge 0$, where $n$ is not necessarily integer here. If $V<E_{0}$ and $V$ is sufficiently close to $E_{0}$, (\ref{for_rho}) and (\ref{for_p}) yield 
\begin{eqnarray}
\rho &\approx & 4 \sqrt{2}\pi C\epsilon^{n+3/2}A_{n}, \\
p&\approx & 2^{3/2}\frac{4}{3}\pi C\epsilon^{n+5/2}B_{n},
\end{eqnarray}
in the lowest order, where 
\begin{equation}
 A_{n}=\frac{\sqrt{\pi}\Gamma(n+1)}{2\Gamma(n+5/2)}, ~~
 B_{n}=\frac{3\sqrt{\pi}\Gamma(n+1)}{4\Gamma(n+7/2)}, ~~
 \epsilon=\frac{E_{0}}{V}-1.
\end{equation}
Thus, we obtain the following polytropic equation of state 
\begin{equation}
p\approx K\rho^{\gamma},
\end{equation}
where 
\begin{equation}
\gamma=\frac{2n+5}{2n+3},~~
K=\frac{2}{3}[4\sqrt{2}\pi C]^{-2/(2n+3)}\frac{B}{A^{(2n+5)/(2n+3)}}
\end{equation}
or $\rho=\tilde{K}p^{\gamma^{-1}}$ with $\tilde{K}=K^{-\gamma^{-1}}$. In this lowest order, we can easily find
\begin{equation}
I\approx -\frac{5\gamma-6}{5\gamma^{2}}\tilde{K}^{2}p^{2(1-\gamma)/\gamma}
\end{equation}
Therefore, if $V$ is sufficiently close to $E_{0}$, or equivalently, $p$ is sufficiently small, $I$ is negative and hence Theorem~\ref{theo_bs} applies for $\gamma>6/5$ or $0\le n<7/2$. It is interesting to note that 
$\gamma=6/5$ is the critical value beyond which the polytrope has a surface of finite radius in Newtonian gravity. Note that the equation of state for the low-pressure regime  cannot be that for isothermal gas.

\subsection{Example of a step-function energy distribution} \label{sec_ex}

In this section and Section \ref{sect_parameters} we construct an explicit family of examples of static solutions of the Einstein-Vlasov system that are spherically symmetric and unique. \par
Let $E_0>0$. We consider the ansatz
\begin{equation} \label{ansatz_step}
F({\bf x}, v) = \Phi(E) = \Theta(E_0 - E),
\end{equation}
where $\Theta$ denotes the Heaviside step function. This ansatz describes a particle distribution where the energy is evenly distributed over the particles up until a cut-off energy $E_0$. \par
We can calculate $\varrho(V)$ and $p(V)$ explicitly from the formulas (\ref{for_rho}) and (\ref{for_p}). We obtain
\begin{align}
\varrho(V) &= \frac{\pi}{2}\left[\frac{E_0}{V} \sqrt{\left(\frac{E_0}{V}\right)^2-1}\left(2\left(\frac{E_0}{V}\right)^2-1\right)-\ln\left(\sqrt{\left(\frac{E_0}{V}\right)^2-1}+\frac{E_0}{V}\right)\right], \\
p(V) &= \frac{\pi}{6}\left[\frac{E_0}{V} \sqrt{\left(\frac{E_0}{V}\right)^2-1}\left(2\left(\frac{E_0}{V}\right)^2-5\right) +3 \ln\left(\sqrt{\left(\frac{E_0}{V}\right)^2-1}+\frac{E_0}{V}\right)\right].
\end{align} 
For the function $I(V)$ an explicit formula can be obtained, too. We have
\begin{equation} \label{for_i_rho_p}
I = \frac{\varrho + p}{5(\varrho + 3p)^2} \left(30 p \frac{\mathrm d\varrho}{\mathrm dp} + (\varrho + 11p)\left(\frac{\mathrm d\varrho}{\mathrm d p}\right)^2 + 5\left(3 p^2 + 4 p \varrho + \varrho^2\right)\frac{\mathrm d^2 \varrho}{\mathrm dp^2}\right).
\end{equation}
We calculate
\begin{align}
\frac{\mathrm d\varrho}{\mathrm dp} &= \frac{\mathrm d\varrho}{\mathrm dV} \frac{\mathrm dV}{\mathrm dp} = \frac{\frac{\mathrm d\varrho}{\mathrm dV}}{\frac{\mathrm dp}{\mathrm dV}} = \frac{3E_0^2}{E_0^2-V^2},\\
\frac{ \mathrm d^2 \varrho}{\mathrm dp^2} &= \frac{1}{\left(\frac{\mathrm dp}{\mathrm dV}\right)^3} \left( \frac{\mathrm d^2\varrho}{\mathrm dV^2}\frac{\mathrm dp}{\mathrm dV} - \frac{\mathrm d\varrho}{\mathrm dV}\frac{\mathrm d^2 p}{\mathrm dV^2} \right) = -\frac{9E_0 V^7}{2 \pi (E_0^2-V^2)^4} \sqrt{\frac{E_0^2}{V^2}-1}.
\end{align}
Inserting into (\ref{for_i_rho_p}) yields
\begin{align}
I(V) &= 4E_0^3 \Bigg(24E_0^7 - 108 E_0^5 V^2 + 139 E_0^3 V^4 - 55 E_0 V^6  \label{long_formula_i} \\
&\qquad\qquad- 5\sqrt{\frac{E_0^2}{V^2}-1}V^5 \left(5V^2-6 E_0^2\right) \ln\left(\sqrt{\frac{E_0^2}{V^2}-1}+\frac{E_0}{V}\right) \Bigg) \nonumber \\ 
&\times \Bigg(5\left(E_0^2 - V^2\right) \Bigg(E_0\sqrt{\frac{E_0^2}{V^2}-1} V \left(2E_0^2 - 3V^2 \right) \nonumber \\
&\hspace{5cm} + V^4 \ln\left(\sqrt{\frac{E_0^2}{V^2}-1}+\frac{E_0}{V}\right)\Bigg)^2\Bigg)^{-1}. \nonumber
\end{align}
Using l'H\^{o}pital's rule we can confirm that $I(V) \to -\infty$, as $V\to E_0$, which corresponds to $p\to 0$ (as stated in Theorem \ref{main_the}). If a particular choice for $E_0$ is made a certain equation of state is fixed and one can calculate the corresponding pressure $p_0$ such that $I(p) \leq 0$ for all $p\leq p_0$ by solving the equation $I(V(p_0))=0$. To illustrate this we consider the example $E_0=0.9$. The resulting functions for $\varrho(V)$ and $p(V)$ can be found in Figure \ref{fig_rhopstep}. The equation of state is illustrated in Figure \ref{fig_eos}. Finally we obtain $I(V)$ as drawn in Figure \ref{fig_step}. 
\begin{figure}[h]
\begin{minipage}{0.49\linewidth}
\begin{center}
\setlength{\unitlength}{1cm}
\begin{picture}(9,5)(0,0)
\put(0,0.225){\includegraphics[width=7.5cm]{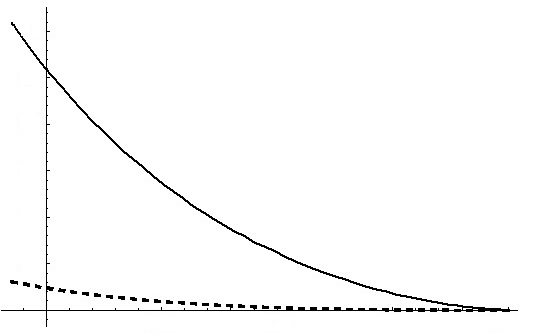}}
\put(0.5, 0.125){$0.7$}
\put(2.075, 0.125){$0.75$}
\put(3.65, 0.125){$0.8$}
\put(5.225, 0.125){$0.85$}
\put(6.8, 0.125){$0.9$}
\put(7.2, 0.525){$V$}
\put(0.2,0.4){$0$}
\put(0.2,1.65){$1$}
\put(0.2,2.9){$2$}
\put(0.2,4.2){$3$}
\put(1.5,1){$p$}
\put(1.5,3.2){$\varrho$}
\end{picture}
\caption{$\varrho(V)$ and $p(V)$ (dashed) for $E_0=0.9$ \label{fig_rhopstep}}
\end{center}
\end{minipage}
\begin{minipage}{0.49\linewidth}
\begin{center}
\setlength{\unitlength}{1cm}
\begin{picture}(9,5)(0,0)
\put(0,0){\includegraphics[width=7.5cm]{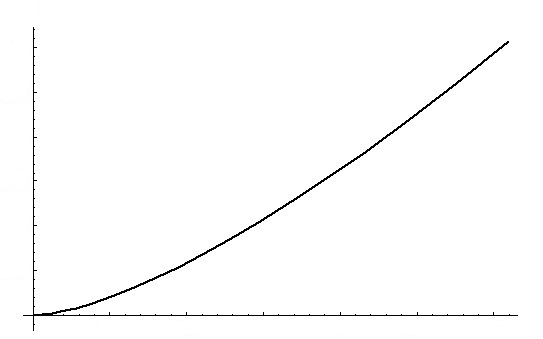}}
\put(0.35, -0.1){$0$}
\put(1.42, -0.1){$0.5$}
\put(2.48, -0.1){$1$}
\put(3.55, -0.1){$1.5$}
\put(4.62, -0.1){$2$}
\put(5.68, -0.1){$2.5$}
\put(6.75, -0.1){$3$}
\put(7.2, 0.25){$\varrho$}
\put(0.1,0.2){$0$}
\put(-0.05,1.43){$0.1$}
\put(-0.05,2.67){$0.2$}
\put(-0.05,3.9){$0.3$}
\put(0.3,4.5){$p$}
\end{picture}
\caption{Equation of state $p(\varrho)$ for the step function ansatz (\ref{ansatz_step}) and $E_0=0.9$ \label{fig_eos}}
\end{center}
\end{minipage}
\end{figure}
\begin{figure}[h]
\begin{center}
\setlength{\unitlength}{1cm}
\begin{picture}(9,5)(0,0)
\put(0,0){\includegraphics[width=9cm]{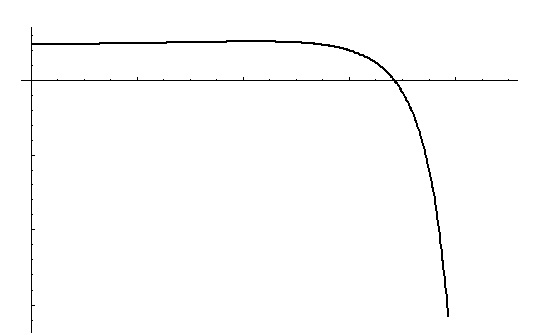}}
\put(0.45, 3.8){$0$}
\put(2.19, 3.8){$0.2$}
\put(3.93, 3.8){$0.4$}
\put(5.66, 3.8){$0.6$}
\put(6.4, 3.8){$V_0$}
\put(7.4, 3.8){$0.8$}
\put(8.8, 4.1){$V$}
\put(-0.3,0.35){$-30$}
\put(-0.3,1.6){$-20$}
\put(-0.3,2.85){$-10$}
\put(0,4.1){$0$}
\put(0.3,5.25){$I(V)$}
\end{picture}
\caption{Plot of $I(V)$ for $E_0=0.9$ \label{fig_step}}
\end{center}
\end{figure}
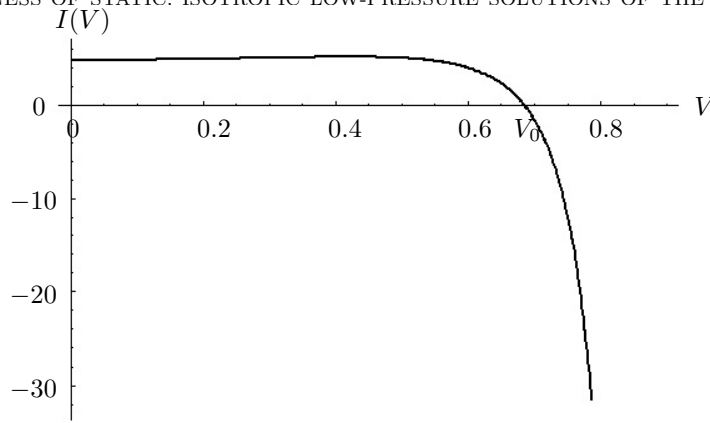
We calculate $V_0 \approx 0.68508$ and the corresponding $p_0 = p(V_0) \approx 0.30645$. As shown in Section \ref{sect_main}, a spherically symmetric static solution of the Einstein-Vlasov system with the ansatz (\ref{ansatz_step}) is unique, if $\sup_x p(x) \leq p_0$. However, at this stage of the analysis it is still open if there exist static solutions corresponding to the ansatz (\ref{ansatz_step}) with cut off energy $E_0 = 0.9$ which satisfy $\sup_{x\in\mathbb R^3}p(x) \leq p_0$. This issue is addressed in the next section.

\subsection{Parameters for unique static solutions} \label{sect_parameters}

Now we construct static solutions of the spherically symmetric Einstein-Vlasov system by numerical means to obtain some insight in the condition that $\sup_{x}p(x)$ must not be too large. For a spherically symmetric regular solution, the space-time manifold is $\mathscr M \cong \mathbb R \times \mathbb R^3$ and can be equipped with coordinates $t\in \mathbb R, r\in[0,\infty), \vartheta\in [0,\pi], \varphi\in[0,2\pi)$ such that the metric reads
\begin{equation}
g=-e^{2\mu(r)} \mathrm dt^2 + e^{2\lambda(r)}\mathrm dr^2 + r^2 \mathrm d\vartheta^2 + r^2 \sin^2\vartheta\, \mathrm d\varphi^2
\end{equation}
for functions $\mu,\lambda:\mathbb R_+\to\mathbb R$. Define the Hawking mass
\begin{equation}
m(r) = 4\pi \int_0^r s^2\varrho(s) \mathrm ds.
\end{equation}
It can be shown that in spherical symmetry the static Einstein-Vlasov system reduces to the integro-differential equation
\begin{subequations}
\begin{eqnarray}
\mu'(r) &=& \frac{1}{1-\frac{2m(r)}{r}}\left(4\pi r p(r) + \frac{m(r)}{r^2}\right), \label{ss1} \\
\mu(0) &=& \mu_c < 0, \label{ss2}
\end{eqnarray}
\end{subequations}
cf.~\cite{rr92}. We observe that $V(r)=e^{\mu(r)}$ is strictly increasing. Since $p(V)$ is decreasing in $V$, cf.~\ref{for_p}), the maximum value will be attained at $V_c := e^{\mu_c}$. So for a spherically symmetric solution we have 
\begin{equation}
\sup_{x\in\mathbb R^3} p(x) = p\left(e^{\mu_c}\right).
\end{equation}
The aim now is to find values for $E_0$ and $\mu_c$ such that the corresponding spherically symmetric solution satisfies $\sup_x p(x) \leq p_0$, thus is unique.\par
A spherically symmetric static solution of the Einstein Vlasov system can be calculated numerically by integrating (\ref{ss1})--(\ref{ss2}) using the methods described in \cite{ar07}. Thereby a solution is (uniquely) determined by the choice of $\mu_c$ and $E_0$. When integrating the function $\mu(r)$ outwards along the radial axis, it will asymptotically approach a fixed value 
\begin{equation}
\mu_\infty := \lim_{r\to\infty} \mu(r).
\end{equation}
An asymptotically flat solution, however, satisfies $\mu_\infty = 0$. After a numerical solution has been constructed, this can be achieved by a rescaling of the time coordinate. Since $E=g(\partial_t,p)$ this in turn affects the particle energy and changes the role of the cut-off energy $E_0$. So $\mu_\infty$, $\mu_c$, and $E_0$ are not independent parameters and the construction of asymptotically flat, spherically symmetric static solutions requires a bit more care. \par
As described for example in \cite{ar07}, it can be done by introducing the variable $y=e^\mu / E_0$. This substitution makes the cut-off energy $E_0$ disappear as a free parameter from the problem, as can be seen as follows. Using the formula (\ref{def_e}) for the energy $E$, we obtain
\begin{equation}
\Phi(E) = \Theta(E_0-E) = \left[1 - \frac{e^{\mu}\sqrt{1+v^2}}{E_0}\right]_+ = \left[1 - y \sqrt{1+v^2}\right]_+.
\end{equation}
The system (\ref{ss1})--(\ref{ss2}) becomes
\begin{subequations}
\begin{eqnarray}
y'(r) &=& \frac{y(r)}{1-\frac{2m(r)}{r}}\left(4\pi r p(r) + \frac{m(r)}{r^2}\right), \label{ss1y} \\
y(0) &=& y_c < 1. \label{ss2y}
\end{eqnarray}
\end{subequations}
A solution in terms of $y$ is then uniquely (as a spherically symmetric solution) determined by the central value $y_c := y(r=0)$. After a solution in terms of $y$ has been constructed, the values of $\mu_c$ and $E_0$ can be calculated, such that $\mu_\infty=0$. \par
In this context we would like to construct a solution with minimal potential $V_c:=e^{\mu_c}$ fulfilling $V_0\leq V_c$, where $V_0$ is the zero of $I(V)$, which is solely determined by the choice of $E_0$ (and $\phi$). For such a solution we have $p(x) \leq p(e^{\mu_c})$ and $I(V) \leq 0$ everywhere. \par
Since, as just described, the values for $\mu_c$ and $E_0$ cannot be chosen directly, we carry out a parameter study, cf.~Figure \ref{fig_e0}. 
\begin{figure}[h]
\begin{center}
\setlength{\unitlength}{0.9cm}
\begin{picture}(15,10)(0,0)
\put(0,0.3){\includegraphics[width=12.6cm]{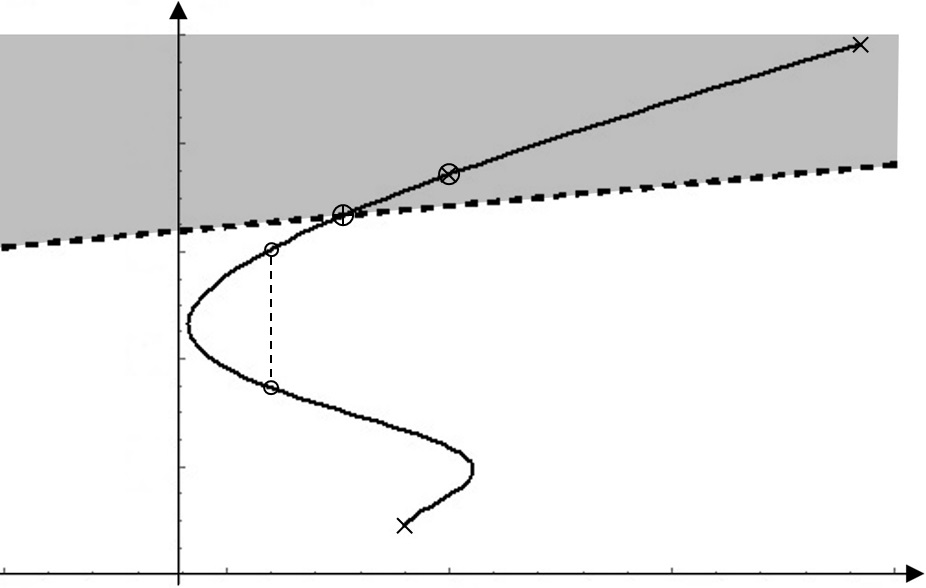}}
\put(-0.2,0){$0.8$}
\put(3.2,0){$0.85$}
\put(6.6,0){$0.9$}
\put(9.9,0){$0.95$}
\put(13.5,0){$1$}
\put(14.1,0.35){$E_0$}
\put(2.1, 0.4){$0$}
\put(2.1, 2.02){$0.2$}
\put(2.1, 3.64){$0.4$}
\put(2.1, 5.26){$0.6$}
\put(2.1, 6.88){$0.8$}
\put(2.1, 8.5){$1$}
\put(2.8, 9){$V$}
\put(13.7,6.55){$V_0(E_0)$}
\put(13.3,8.4){$y_c=0.99$}
\put(6.5,1.1){$y_c=0.1$}
\put(4.3,4){two solutions with the same surface potential $E_0$}
\put(3.6,7.5){low-pressure regime}
\put(5.3,5.5){intersection point, $\hat y_c = 0.76$}
\put(5.3, 5.05){$(\hat E_0 \approx 0.875, \hat V_c \approx 0.665)$}
\put(7.08,6.42){example of Sect.~5.2}
\end{picture}
\caption{For each choice of $E_0$ the function $I(V)$ has one zero which is called $V_0(E_0)$. The tuples $(E_0,V_0)$ lie on the dashed line. The continuous line represents a succession of central values $y_c$ for $y(r)$ between $0.1$ and $0.99$ (marked with $\times$), and the resulting tuples $(E_0, V_c)$, where $V_c=e^{\mu_c}$ is the (minimal) potential at the center. The solution at the intersection is marked with $\oplus$ and denoted with a hat. The solutions which has been discussed as an example in Sections \ref{sec_ex} and \ref{sect_parameters} above is marked with a $\otimes$. 
 \label{fig_e0}}
\end{center}
\end{figure}
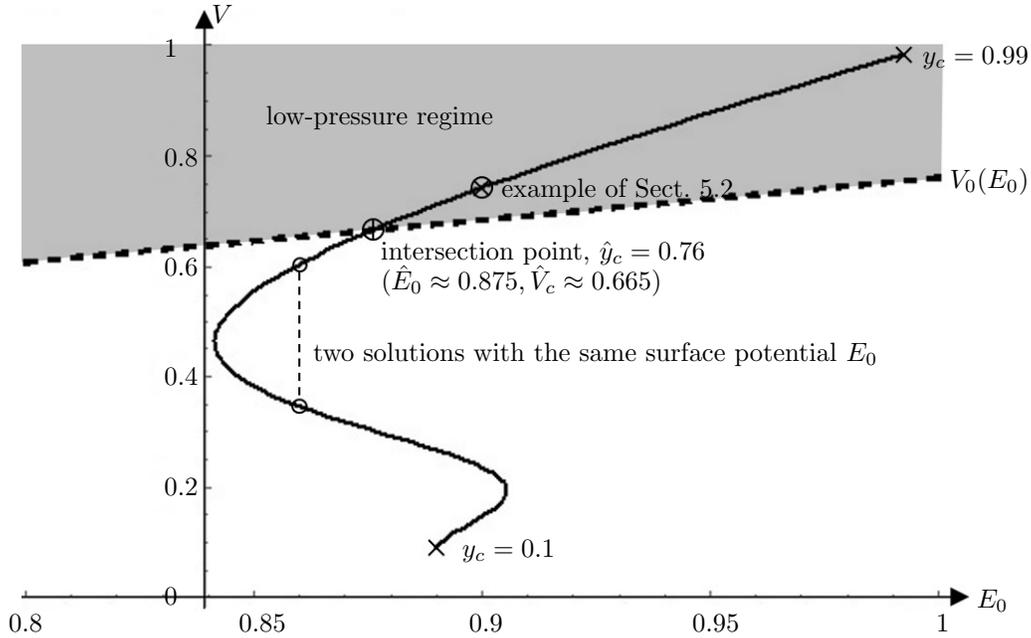
First, for each value of $E_0$, we calculate $V_0$ by solving
$I(V_0)=0$, where $I(V)$ is given in (\ref{long_formula_i}). We obtain the dashed line. Then we calculate steady states for a succession of central values $y_c$ between $0.1$ and $0.99$, by integrating (\ref{ss1y})--(\ref{ss2y}). Linking the $\left( E_0,e^{\mu_c}\right)$-points corresponding to these solutions one obtains the continuous line in Figure \ref{fig_e0}. The part of the continuous line lying above the dashed line corresponds to solutions with $V_0\leq e^{\mu_c}$, lying in the low-pressure regime. The intersection (marked with a ``$\oplus$'' in Figure \ref{fig_e0}) lies approximately at $\hat y_c=0.76$, corresponding to $\hat E_0\approx 0.875$, $\hat V_c\approx 0.665$. In order to compare with the example, depicted in Figures \ref{fig_rhopstep} and \ref{fig_eos}, and we consider a solution corresponding to $y_c=0.83$, $E_0\approx 0.902$, $V_c\approx 0.748$ (marked with an ``$\otimes$'' in Figure \ref{fig_e0}). If we compare this value for $V_c$ to $V_0 \approx 0.685$ that we have calculated above for $E_0=0.9$ we see that this particular spherically symmetric solution is in the realm where the uniqueness theorem (Theorem \ref{main_the}) is valid. \par
Figure \ref{fig_e0} contains two messages. First, that there indeed exist static solutions of the Einstein-Vlasov system in the low-pressure regime, i.e.~solutions that the uniqueness theorem applies to. Second, it gives some indication that uniqueness does not hold for arbitrary pressures. Observe that the continuous line takes no turns above the dashed line (in the low-pressure regime), whereas it oscillates below the dashed line. If the continuous line oscillates, this means that there are several spherically symmetric solutions corresponding to the same surface potential $E_0$. One example ($E_0\approx 0.86$) is marked by ``$\circ$'' in Figure \ref{fig_e0}. So in the realm below the dashed line the solutions are not unique, not even under the restriction of spherical symmetry. The question whether or not there also exist several solutions with the same surface potential, which are not spherically symmetric, can of course not be answered with this parameters study. Furthermore for some solutions in the low-pressure regime there exist solutions with higher pressures that have the same surface potential. For example, for $E_{0}\approx 0.902$, we can see that there are two solutions below the dashed line with $y_{c}$ between $0.1$ and $0.76$ besides the unique solution marked with an ``$\otimes$'' above the dashed line. So, in conclusion, uniqueness seems to hold only {\em within} the low-pressure regime. In other words, the uniqueness theorem in the current paper assures that for a given value of surface potential, the low-pressure static solution in an asymptotically flat spacetime is spherically symmetric and uniquely determined, without excluding the existence of high-pressure solutions which may be or may not be spherically symmetric.

\subsection{Comparison to astrophysical objects}

The concentration parameter 
\begin{equation}
\Gamma = \sup_{r\in(0,\infty)} \frac{2m(r)}{r}
\end{equation}
is a dimensionless quantity indicating how relativistic the solution is. The solution is very relativistic if $\Gamma$ is large. As a generalization to the well-known Buchdahl inequality, the upper bound $\Gamma \leq \frac 8 9$ has been shown for a large class of matter models including Vlasov matter and the perfect fluid model \cite{a08}. Moreover, the smaller $\mu_c<0$ is chosen (or equivalently the smaller $y_c$ is chosen, cf.~Figure \ref{fig_e0}) the bigger the value of $\Gamma$ will be. This has been made precise in \cite{a07}. \par
Let $\hat \Gamma$ be the concentration parameter of the solution corresponding to $\hat y_c$, the intersection point in Figure \ref{fig_e0}. Then, if another solution has $\Gamma < \hat\Gamma$, its maximal (central) potential $V_c$ will be larger than the maximal potential $\hat V_c$ of the critical solution. This means it will be unique. We calculate $\hat \Gamma \approx 0.292$. It might be instructive to set this number into relation to values of $\Gamma$ for other objects in the universe. \par
Taking into account physical units, at the surface of a spherical object
with mass $M$ and radius $R$ we have $\Gamma=2MG/(c^2r)$, where $G$ is
the gravitational constant and $c$ is the speed of light. A very extreme
situation is the surface of a neutron star. The survey article
\cite{of16} suggests that for a model calculation one can assume
$M=3\cdot 10^{30}\,\mathrm{kg}$ and $R=10\,\mathrm{km}$. This yields
$\Gamma \approx 0.4424$. A value clearly larger than $\hat \Gamma$ for
the family with ansatz function given in (\ref{ansatz_step}). So Theorem \ref{theo_bs} is not
suitable in the regime of strong gravity like at the surface of a neutron star. The authors of \cite{bs90},  where Theorem \ref{theo_bs} emanates from, make a similar remark. \par
Taking radius and mass of the sun however, we calculate $\Gamma \approx 4.24 \cdot 10^{-6}$ at the surface of the sun. This is clearly smaller than $\hat \Gamma$. As application of Theorem \ref{main_the} we have globular clusters in mind. Since the typical mass and size of a globular cluster are hundreds of thousands solar masses and several parsecs, respectively, the ratio $\Gamma$ is approximately $10^{-8}$. This means that if the distribution function is a decreasing function of $E$ up to the cut off $E_0 < \infty$, Theorem \ref{theo_bs} applies to globular clusters.  \par
Finally we remark that there might be a connection between stability and the fact that a solution necessarily is spherically symmetric, like in the non-relativistic case. Even though only very little is known about the question of stability of static solutions of the Einstein Vlasov system, it is conjectured that solutions with small values of $\Gamma$ are  stable, whereas highly relativistic solutions with large values of $\Gamma$ are conjectured to be unstable. The reader is referred to \cite{ar06} for a numerical study of this question.

\end{document}